\documentclass[conference,10pt]{IEEEtran}
\usepackage[latin1]{inputenc}
\usepackage{amsthm,amsfonts,amssymb,amsmath}
\usepackage{ae}
\usepackage{color}

\usepackage{cite}
\usepackage{flushend}

\newtheorem{thm}{Theorem}[section]

\theoremstyle{definition}

\newcommand{\Z}{\mathbb Z}

\newcommand{\Qc}{\mathcal Q}
\newcommand{\C}{\mathbb C}

\newcommand{\E}{\mathcal E}

\newcommand{\cli}{\mathcal{C}\ell}        

\newcommand{\ket}[1]{\left\vert #1 \right>}
\newcommand{\bra}[1]{\left< #1 \right\vert}

\newcommand{\gera}[1]{\langle #1 \rangle}

\title{Decoder for Nonbinary CWS Quantum Codes}

\author{\IEEEauthorblockN{Nolmar Melo}
\IEEEauthorblockA{Laborat\'{o}rio Nacional de Computa\c{c}\~{a}o Cient\'{\i}fica\\
Petr\'{o}polis, RJ 25651-075, Brazil \\
Email: nolmar@lncc.br}
\and
\IEEEauthorblockN{Douglas F.G. Santiago}
\IEEEauthorblockA{Universidade Federal dos Vales do Jequitinhonha e Mucuri\\
Diamantina, MG 39100000, Brazil \\
Email: douglassant@gmail.com}
\and
\IEEEauthorblockN{Renato Portugal}
\IEEEauthorblockA{Laborat\'{o}rio Nacional de Computa\c{c}\~{a}o Cient\'{\i}fica\\
Petr\'{o}polis, RJ 25651-075, Brazil \\
Email: portugal@lncc.br}}

\begin{document}

\maketitle

\begin{abstract}
We present a decoder for nonbinary CWS quantum codes using the structure of union codes. The decoder runs in two steps: first, we use a union of stabilizer codes to detect a sequence of errors, and second, we build a new code, called union code, that allows the error correction.
\end{abstract}

\section{Introduction}

Quantum computers is able to solve many hard problems in polynomial time and to increase the speed of most algorithms~\cite{Shor1994, Grover1996, Mos09, CW10}. Decoherence problems are inherent in these computers requiring the use of quantum error correcting codes (QECCs)~\cite{Calderbank1996, Steane1996, Laflamme1996, Gottesman1996}.

A large class of good binary codes is known in literature~\cite{Knill1996, Grassl2009, Grassl2009a , Beigi2011}. However, in order to build a quantum fault tolerant quantum computer, concatenation of quantum codes plays a crucial role. The optimum concatenation is obtained when nonbinary codes are used~\cite{Beigi2011}.

An important class of nonadditive codes, called CWS, has been studied recently~\cite{Chuang2009, Chen2008, Cross2008, Smolin2007, Yu2008}. The framework of CWS codes generalizes the stabilizer code formalism and was used to build some good nonadditive codes. The codification of binary and nonbinary CWS codes is well known, whereas the decodification is only known for binary CWS codes~\cite{Yu2008, Li2010}. In this paper, we present an algorithm to decode nonbinary CWS codes, generalizing the procedure described in Ref.~\cite{Li2010}.


This article is divided in the following parts: In Section II, the CWS codes are briefly reviewed. In Section III, the theory of \textit{union codes} are presented. Such codes provide the basis for our decoder. In Section IV, measurement operators for union codes are presented. In Section V, the main theorems are proved. In Section VI, the nonbinary decoder is presented and an analysis of the computational cost is performed. In Section VII, an example is worked out. In Section VIII, the conclusions are presented.

\section{CWS Codes}

Nonbinary CWS codes use the generalized Pauli group $G_d$ over qudits~\cite{Ashikhmin2001, Ketkar2006, Hu2008, Chen2008}. Let $\chi$ be the character of group $\Z_d$ in the unit cycle of $\C$. By definition, $\chi(x)=\exp(\frac{i2\pi x}{d})$. We denote $\omega=\chi(1)$. The action of $X$ over qudits is $X\ket i = \ket{i+1}$ and the action of $Z$, which performs a phase shift, is $Z\ket{i}=\omega^{i}\ket{i}$. In matrix form, we have
\[
         X=\left(
                  \begin{array}{c|c}
                           0 & 1\\
                           \hline
                           I & 0
                  \end{array}
          \right)_{d\times d},\;\;
         Z=\left(
                  \begin{array}{ccccc}
                           1 & 0 & 0 & \dots & 0 \\
                           0 & \omega & 0 & \dots & 0 \\
                           0 & 0 & \omega^2 & \dots & 0 \\
                           \vdots & \vdots & \vdots & \ddots & 0 \\
                           0 & 0 & 0 & \dots & \omega^{d-1}
                  \end{array}
         \right).
\]
It is easy to check that $X^d=Z^d=I$ and $ZX=\omega XZ$. Pauli group $G_d$ is generated by $\{X, Z,\omega^jI\}$, where $j=0,\dots ,d-1$. The Pauli group acting on $n$ qudits is $G_d^n = \otimes^n G_d$.

A nonbinary CWS code is a nontrivial vector subspace of $\C^{d^n}$. If this subspace has dimension $K$, the code is denoted by $((n,K))_d$ and, if $\delta$ is the minimum distance, the notation is $((n,K,\delta))_d$. A CWS code is described by a subgroup $S$ of the Pauli group (called stabilizer group) and a set of $K$ Pauli operators $W=\{w_l\}_{l=1}^K$ called word operators.

Group $S$ stabilizes a single word, usually $\ket{S}$ and, in the binary case, we have $S=\gera{g_i,\dots,g_m}$ with $m=n$, whereas in the nonbinary case we have $m\geq n$. The generators of $S$ have the form $g_i=X^{r_i}Z^{t_i}$, where $r_i$ and $t_i$ are vectors with entries in $\Z_d$. We can build a matrix $[r|t]$ of dimensions $m\times 2n$, which is useful for establishing a connection with a classical code.

A basis for the quantum code is $\{w_l\ket{S}; \; w_l \in W \}$, where $\ket S$ is stabilized by group $S$. Note that for each $w_i\in W$, $w_l\ket{S}$ is stabilized by $w_lSw_l^{\dagger}$. Moreover, for $g_k\in S$, we have $w_lg_kw_l^\dagger= \omega^{l_k}g_k$ $\ket{w_l}$, so a classical vector on $\Z_d$
\[c_l = (l_1, \dots, l_m)\]
can be associated to $w_l$.

The error correction conditions for quantum codes state that, in order to detect a set of errors $\E$, it is necessary and sufficient that
\[
         \bra{\psi_i}E\ket{\psi_j} = C_E \delta_{ij},
\]
for all $E\in \E$, where $\ket{\psi_i}$ and $\ket{\psi_j}$ are in an orthonormal basis of the code. Note that $C_E$ does not depend on $i$ and $j$. $E_1, E_2\in \E$ are correctable if and only if
\[
          \bra{\psi_i}E_1^\dagger E_2\ket{\psi_j} = C_{E_1E_2} \delta_{ij}.
\]
Again, $C_{E_1E_2}$ does not depend on $i$ and $j$. An error is degenerate if $C_E\neq 0$. Two distinct errors, $E_1,E_2\in \E$, belong to the same degeneracy class when their actions on the code are the same, that is, $C_{E_1E_2}\neq 0$. A code is said degenerate when the error set has a degenerate element~\cite{Knill1997}.

It is enough to consider errors as operators in the Pauli group acting on the code, that is, an error $E$ has the form $\alpha Z^vX^u$, where $\alpha\in \C$ and $v,u\in \Z_{d}^n$. We can map errors to classical vectors using function $\cli_S$ given by
\begin{equation}
         \cli_S(E) = \sum_{l=1}^n v_lr_l - u_lt_l,
\end{equation}
where $r_l$ and $t_l$ are the columns of the matrix $[r|t]$. An error $E$ is detectable in the quantum code if and only if $\cli_S(E)$ is detectable in the associated classical code and $\cli_S(E)\neq 0$ or $\forall l,\; w_lE=Ew_l$. An error is degenerate when $\cli_S(E)=0$ and two distinct errors, $E_1,E_2\in \E$, belong to the same degeneracy class when $\cli_S(E_1)=\cli_S(E_2)$~\cite{Cross2008}.

\section{Union and USt Codes}

Let $S$ be the stabilizer group of code $C=[[n,k,d']]_d$. Let $\mathcal{C}_C$ be the centralizer of $S$ and $\mathcal{T}$ be a subset of a transversal set of $G_d^n/\mathcal{C}_C$. Then, the set
\begin{equation}
         \mathcal{Q} = \bigoplus_{t\in \mathcal{T}}tC
\end{equation}
is a quantum code with parameters $((n,Kd^k,d''))_d$, where $K=\#\mathcal{T}$ and $d'' \leq d'$.

Note that if $t_1,t_2\in \mathcal{T}$ are distinct, then $t_1C\perp t_2C$. In fact, $t_1^\dagger t_2\notin\mathcal{C}_C $, therefore there exists $s\in S$ such that $st_1^\dagger t_2 = \alpha t_1^\dagger t_2 s$ and $\alpha\neq 1$, that obeys $\bra{i}t_1^\dagger t_2\ket{j} =0$, where $\ket{i},\ket{j}\in C$. If $B=\{\ket{w_i}\}$ is a base for $C$, then $\bigcup_{t\in  \mathcal{T}}tB$ is a basis for $\mathcal{Q}$. $\mathcal{Q}$ is called Union Stabilizer Code(USt)~\cite{Grassl2008a, Grassl2008, Li2010, Li2010a}.

More general yet is what we call a quantum union code. Let $C_1$ and $C_2$ be two quantum codes with parameters $((n,K_1,d_1))$ and $((n,K_2,d_2))$, respectively. Let $B_1$ ($B_2$) be an orthogonal basis of $C_1$ ($C_2$). Suppose that $C_1\perp C_2$, then $B=B_1 \bigcup B_2$ also is a basis for a vector space, which is the union code. Note that the union code is given by $C=C_1\oplus C_2.$ Note that a USt code is also a union code~\cite{Grassl1997}.


\section{Projectors on Union Codes}

In this section, we show how to find the projector of a union code. Let $ M $, $P(M)$ and $P_M$ be a measurement operator, the space stabilized by $M$, and the orthogonal projector on $P(M)$, respectively. We have $M =2P_M-I$. We also use the notation $P_Q$ for the projector of a generic code $ Q $.

Let $M_Q$ be the measurement operator of code $Q$. If $\{\ket{w_1}, \dots, \ket {w_k} \}$ is an orthogonal basis of the code, then
\[
	 P_Q= \sum_{i=1}^k \ket{w_i}\bra{w_i},
\]
and
\begin{eqnarray*}
   M_Q &=& 2\sum_{i=1}^k \ket{w_i}\bra{w_i} -I \\
       &=& -\prod_{i=1}^k (I -2 \ket{w_i}\bra{w_i}).
\end{eqnarray*}

Suppose that $Q$ is a union code $Q=C_1\oplus C_2$ and let $P_1$ and $P_2$ be the projectors on codes $C_1$ and $C_2$, respectively. Since $C_1\perp C_2$, the projector on $Q$ is $P=P_1\oplus P_2$ and the measurement operator is $M_Q = 2P_1\oplus P_2 -I$.

Suppose that the union code $Q$ has the form
\begin{equation}\label{code:1}
	 \mathcal{Q}=\bigoplus_{t\in T} tQ_0,
\end{equation}

where $Q_0$ is a quantum code. How does one find the code measurement operator using operator $M_0$ of code $Q_0$? We answer this question below.

Let \(M_1\) and \(M_2\) be two commutative measurement operators, we define the measurement operator \(M_1\wedge M_2\) as the operator that stabilize the space \(P(M_1)\bigcap P(M_2)\). We have that the projector associated to this operator satisfies \(P_{M_1\wedge M_2} = P_{M_1}P_{M_2}\).

Let $A$ and $B$ be two vector subspaces of $\C^m$. Define the following associative operation:
\begin{equation}
         A\triangle B = (A\cap B^\perp )\oplus (A^\perp \cap B).
\end{equation}
Let $M_1$ and $M_2$ be two commutative measurement operators. The measurement operator associated with the space $P(M_1)\triangle P(M_2)$ is denoted by $M_1\boxplus M_2$, that is
\begin{eqnarray}
         P(M_1\boxplus M_2) & = & P(M_1)\triangle P(M_2) \nonumber\\
         &=&(P(M_1)\cap P(M_2)^\perp)\oplus \nonumber \\
         & & (P(M_1)^\perp\cap P(M_2)).
\end{eqnarray}
Note that $P_{M_1\boxplus M_2} = P_{M_1}(I - P_{M_2}^\perp) +  P_{M_2}(I - P_{M_1}^\perp)$ and
\begin{eqnarray}
         M_1\boxplus M_2 & = & 2P_{M_1\boxplus M_2} - I \nonumber \\
         & = & 2P_{M_1}(I - P_{M_2}^\perp)+\nonumber \\
         & &  \,  P_{M_2}(I - P_{M_1}^\perp)) - I \nonumber \\
         & = & -(2P_{M_1} - I)(2P_{M_2}-I )\nonumber \\
         & = & - M_1M_2.
\end{eqnarray}

Consider again code $\mathcal Q$ of equation (\ref{code:1}). Let $Q_0$ be the stabilizer and $S=\{G_i\}_{i=1}^k$ is a set of stabilizers of code $Q_0$. Define $Q_t = tQ_0$ and $M_{it}=tG_it^\dagger$. Note that $M_{it}$ stabilizes $t\ket{w}$, for $\ket{w}\in Q_0$. We have $Q_t=\bigcap_{i=1}^k P(M_{it})$. Then
\begin{equation}
         \mathcal{Q} = \bigoplus_{t\in T}\bigcap_{i=1}^k P(M_it) =    \mathop{\bigtriangleup}_{t\in T}\bigcap_{i=1}^k P(M_{it})
\end{equation}
and the associated measurement operator is
\begin{equation}
         M_{\mathcal{Q}} = \mathop\boxplus_{t\in T}\bigwedge_{i=1}^kM_{it}.
\end{equation}

When code $Q_0$ is not additive (stabilizer), we use the classical way to build the projector for a vector space employing an orthonormal basis. Let $B=\{\ket{w_i} \}_{i=1}^l $ be an orthonormal basis of $Q_0$. The set $\bigcup_{t\in T}tB$ is an orthonormal basis for the code $\mathcal{Q}$. Therefore, the projector of $\mathcal{Q}$ is
\begin{equation}
         P_\Qc = \sum_{t\in T}\sum_{i=1}^l t\ket{w_i}\bra{w_i}t^\dagger.
\end{equation}

\section{Measurements on Union and USt Codes}

Let $\E$ be a set of correctable errors of a quantum code and $D\subset\E$. Define the nondegenerate complement of $D$ in $\E$ as the set
\begin{equation}
         \E_{\overline{D}} = \{E\in \E ; C_{EF}=0,\;\; \forall F\in D \}.
\end{equation}
When $\E$ is nondegenerate, $\E_{\overline{D}}$ is exactly the complement of $D$, that is, $\E_{\overline{D}}=\E\setminus D$. Following we use the notation \(D(Q)=\{h\ket{\psi}; h\in D,\; \ket{\psi}\in Q\}\).
%
%

\begin{thm}\label{detecta0}
         Let $Q$ be a CWS code and $D$ a finite commutative group of correctable errors. Then $Q_{D}=D(Q)$ is a USt code.
\end{thm}
\begin{proof}
         Let $ B=\{\ket{w_1},\dots , \ket{w_k}\} $ be a basis for code $Q$. Then $D(B) = \bigcup_{\alpha\in D} \alpha(B) =\bigcup_{\alpha\in D}\bigcup_{i=1}^k \alpha\ket{w_i} = \bigcup_{i=1}^k D(\ket{w_i})$. Note that $D(\ket{w_i})$ is an additive CWS code, that is, a stabilizer code. Now we will show that $D(\ket{w_i})$ and $D(\ket{w_j})$ are mutually orthogonal, for $i\neq j$.

         Let $\alpha_1, \alpha_2 \in D$. Then $\alpha_1^\dagger \alpha_2$ a detectable error  and the error correction conditions $ \bra{w_i}\alpha_1^\dagger  \alpha_2 \ket{w_j} = 0 $ imply that $D(\ket{w_i})$ is orthogonal to $D(\ket{w_j})$.

         We have $\ket{w_i}=w_i\ket 0$ for word operators $w_i$, and $w_id_j = \alpha_{ij}d_jw_i$, where $d_j\in D$ and $\alpha_{ij}\in \C$. Let $D_i = \{\alpha_{ij}d_j ; d_j\in D \}$. We have that the codes generated by $D({\ket{0}})$ and by $D_i(\ket{0})$ are the same. Then $\bigcup_{i=1}^k D(\ket{w_i}) = \bigcup_{i=1}^k w_iD_i(\ket{0}) $.

         Let $Q_0$ be the code generated by $D(\ket 0)$. We have
\[
                  D(Q) = \bigoplus_{i=1}^k w_iQ_0,
\]
         that is, $D(Q)$ is USt.
\end{proof}


\begin{thm} \label{detecta}
         Let $Q$ be a CWS code defined from a stabilizer set $S$ and by a classical code $C$. Let $D$ be a commutative group, the elements of which are in the set of correctable errors $\E$. Then the USt code $Q_{D}=D(Q)$ detects all errors in $\E_{\overline D}$.
\end{thm}
\begin{proof}
         Let $E\in \E_{\overline{D}}$, $\alpha_1,\alpha_2\in D$ and $\ket{w_i},\ket{w_j}\in Q$ such that $\alpha_1\ket{w_i}\perp \alpha_2\ket{w_j}$. Then $\alpha_1^\dagger \alpha_2 \in D\subset \E$ and
\[
                  \bra{w_i}\alpha_1^\dagger E \alpha_2\ket{w_j} = a\bra{w_i}\alpha_1^\dagger  \alpha_2 E\ket{w_j} = 0,\;\;\;\ a\in \C,
\]
         that is, $E$ is detectable.
\end{proof}

We have described how to detect errors in USt codes for the nonbinary case, generalizing the method given in Ref.~\cite{Li2010} for binary codes. When working with the latter case, it is enough to use the criteria defined in Theo.~\ref{detecta}. However, for nonbinary codes, the above theorem is not enough. In order to find the error in this case, which is the main contribution of this paper, we make use of union codes.
%


\begin{thm} \label{detecta:2}
         Let $Q$ be a CWS code and $D$ a commutative group. $D\setminus \{I\}$ is a set of nondegenerate and correctable errors. Let $D_l\subset D$ be the set $\{d_{l_1}^{k_1}\cdots d_{l_t}^{k_t} ; 0<k_i<d\; \forall i \}$, where $d_{l_i}$ are $t$ independent generators of $D$. If $D_{l_s}^r = D_l \setminus \{d_{l_1}^{k_1}\cdots d_{l_s}^r\cdots d_{l_t}^{k_t} ; 0<k_i<d\; \forall i\neq s \}$, then the code $\mathbb{D}(Q)$ where $\mathbb{D}=D_l\setminus D_{l_s}^r$ detects $\E = D_{l_s}^r $.
\end{thm}

\begin{proof}
         Let \(\alpha_1, \alpha_2 \in\mathbb{D}  \), we have that \(\alpha_1^\dagger \alpha_2  = a_1d_{l_1}^{k_1}\cdots d_{l_s}^0\cdots d_{l_t}^{k_t} \), \(a_1\in \C\). Now let \(E\in \E\), \(\alpha_1^\dagger   \alpha_2 E = d_{l_1}^{h_1}\cdots d_{l_s}^{h_s}\cdots d_{l_t}^{h_t} \)  with \(h_s\neq 0\). Then \(\alpha_1^\dagger   \alpha_2 E \in D\setminus \{I\}\) ie \(\alpha_1^\dagger   \alpha_2 E\) is a nondegenerate error of code \(Q\). So
         \[
                  \bra{w_i}\alpha_1^\dagger E  \alpha_2 \ket{w_j} = \bra{w_i}\alpha_1^\dagger   \alpha_2 E \ket{w_j} = C_{\alpha_1^\dagger \alpha_2 E}\delta_{ij} = 0,
         \]
         where \(\ket{w_i},\ket{w_j}\) belongs to a orthonormal basis of code \(Q\),  in other words, $\mathbb{D}(Q)$ detect \(E\).
\end{proof}


\section{Decoding}

In the previous section, we have showed that, given a quantum code and a set of correctable errors, we can build another code, which includes the first, detecting the errors. However, we want to correct errors and not only to detect them.

Suppose that the set of correctable errors $\E$ can be decomposed into a union of finite abelian groups and nondegenerate elements, that is, $\E = \bigcup_{j=1}^t D_j$. Using Theo.~\ref{detecta}, we can find the group that contains the error. To correct the error we use the following strategy: Suppose that the error is in group $D_j$ and it has $m$ generators, that is, $D_j=\langle d_1 ,\dots ,d_m \rangle$. Define $D_j^l$ as the group with a generator less, that is, $D_j^l=\langle d_1,\dots, \widehat{d_l},\dots ,d_m \rangle$. By performing $m$ measurements, we find out that the error has the form
\begin{equation}
         E= d_{i_1}^{k_1}\cdots d_{i_t}^{k_t},
\end{equation}
where $t\leq m$ and $0< k_i <d-1$. It remains to find the value of $k_i$.

We use Theo.~\ref{detecta:2}. Take $D_j$ to be the group, and $D_l$ the subset of $D_j$, with all errors found above. For all $s\in \{1,\dots,t\}$ and $r\in\{1,\dots, d-1\}$, we perform the error detection in code $D_{l_s}^r(Q)$. This is the last step to find the error.

\section{Example}

We will use the decoder above described in the family of codes \(((5,d,3))_d\) with \(d>3\) presented in Ref.~\cite{Hu2008}, which are nonadditive CWS codes. This family is described by the stabilizer group
\[S=\gera{X_1Z_2Z_5, Z_1X_2Z_3, Z_2X_3Z_4, Z_3X_4Z_5, Z_1Z_4X_5}\]
and by the word operators set \(W=\{Z^{v_j},Z^a,Z^b\}\), where \(v_j=(j,j,j,j,j)\), for \(j\not\in \{2,d-1\}\), \(a=(2,-1,-1,2,-1)\) and \(b=(-1,2,2,-1,2)\). Those codes correct all weight-1 quantum errors. Let \(\E\) be the set of all weight-1 errors including \(I\).

To prove that this code can correct every weight-1 error, we use function \(\cli_S\) and the classical code associated to set of word operators. The classical code is given by \(C=\{v_j,a,b\}\) and
\begin{eqnarray*}
\cli_S(\E)& = &\{(r,0,0,0,0), (0,r,0,0,0), (0,0,r,0,0),\\
&&(0,0,0,r,0), (0,0,0,0,r), (r,0,0,r,0),\\
&& (0,r,0,0,r), (r,0,r,0,0), (0,r,0,r,0),\\
&&(0,0,r,0,r), (0,0,r,h,r), (r,0,0,r,h),\\
&&(h,r,0,0,r), (r,h,r,0,0), (0,r,h,r,0); \\
&& 0\leq r < d ,\;\; 0<h<d\}.
\end{eqnarray*}
Note that the difference between any two vectors in \(C\) can not be equal to the sum of any two elements of \(\cli_S(\E)\). This shows that \(C\) can correct weight-1 errors.

%
%
%
%
%


%
%
%
%

As a consequence of the graph structure of the stabilizer, every nonbinary Pauli error acting on $\ket{0}$ can be equivalently replaced by some qudit phase flip errors~\cite{Hu2008}. We may consider all the word operators $w_i$ in the format $w_i=Z^{c_i}$.

Now, write $\E=\bigcup_{i=1}^5 D_i$, where $D_i$ are the groups generated by $\{Z^{Cl_{S}(X_i)},\,Z^{Cl_{S}(Z_i)}\}$, $1\le i\le 5$. These groups satisfy Theo.~\ref{detecta0} and~\ref{detecta}. So, the decoder can be applied and with at most $5-1=4$ USt measurements in order to detect the group error $D_i$, that is, to locate the error. After these measurements, we will perform two more USt measurements to determine if some of the generators of $D_i$ are missing in the expression of the error. Then, we construct the groups $D_{i}^{1}=\langle Z^{Cl_{S}(X_i)} \rangle$ and $D_{i}^{2}=\langle Z^{Cl_{S}(Z_i)} \rangle$, perform measurements according to the USts codes $D_{i}^{1}(Q)$ and $D_{i}^{2}(Q)$. We can find out, for example, in the worst case, that no generator of $D_i$ is missing in the expression of the error. To obtain explicitly the power associated with the generators, we use Theo.~\ref{detecta:2}. For each fixed $r_{1}$, $0<r_{1}<d$, consider the sets $D_{r}^{1}=\{(Z^{Cl_{S}(X_i)})^{r_{1}} (Z^{Cl_{S}(Z_i)})^{k}\}$, $0<k<d$, and the union code $D_{r}^{1}(Q)$. Performing a measurement in this union code, we obtain power $r_1$ of the first generator. To obtain the second power, we repeat the procedure with the set $D_{r}^{1}=\{(Z^{Cl_{S}(X_i)})^{k} (Z^{Cl_{S}(Z_i)})^{r_{2}}\}$, $0<k<d$, obtaining the expression of the error $E=(Z^{Cl_{S}(X_i)})^{r_{1}} (Z^{Cl_{S}(Z_i)})^{r_{2}}$. The number of union code measurements is $2d-2$, because we perform $d-1$ measurements to obtain power $r_{1}$ and $d-1$ measurements for power $r_{2}$.

\section{Conclusion}

The formalism of CWS codes is a procedure to find both, additive and  nonadditive codes, which generalizes the formalism of stabilizer codes. It was used to build some optimal nonadditive codes, such as codes ((9,12,3)) and ((10,24,3))~\cite{Yu2008}. A decoding procedure for binary codes of this class was described recently~\cite{Li2010}.

We have described a decoding procedure for the nonbinary case. Part of the procedure is a straightforward generalization of the binary case, using union of stabilizer codes (USt). In the binary case, dealing with USts codes is enough, whereas in the nonbinary case, after finding the group error $D$ and $D_l\subset D=\{d_{l_1}^{k_1}\cdots d_{l_t}^{k_t} ; 0<k_i<d\; \forall i \}$, where $d_{l_i}$ are $t$ independent generators of $D$, we have to use union code measurements to obtain $k_i$ using Theo.~\ref{detecta:2}.

\section*{Acknowledgement}

We thank CNPq's financial support.



\bibliography{quantum-codes}

\begin{thebibliography}{10}
\providecommand{\url}[1]{#1}
\csname url@samestyle\endcsname
\providecommand{\newblock}{\relax}
\providecommand{\bibinfo}[2]{#2}
\providecommand{\BIBentrySTDinterwordspacing}{\spaceskip=0pt\relax}
\providecommand{\BIBentryALTinterwordstretchfactor}{4}
\providecommand{\BIBentryALTinterwordspacing}{\spaceskip=\fontdimen2\font plus
\BIBentryALTinterwordstretchfactor\fontdimen3\font minus
  \fontdimen4\font\relax}
\providecommand{\BIBforeignlanguage}[2]{{%
\expandafter\ifx\csname l@#1\endcsname\relax
\typeout{** WARNING: IEEEtran.bst: No hyphenation pattern has been}%
\typeout{** loaded for the language `#1'. Using the pattern for}%
\typeout{** the default language instead.}%
\else
\language=\csname l@#1\endcsname
\fi
#2}}
\providecommand{\BIBdecl}{\relax}
\BIBdecl

\bibitem{Shor1994}
P.~Shor, ``Algorithms for quantum computation: discrete logarithms and
  factoring,'' in \emph{Foundations of Computer Science, 1994 Proceedings.,
  35th Annual Symposium on}, nov 1994, pp. 124 --134.

\bibitem{Grover1996}
\BIBentryALTinterwordspacing
L.~K. Grover, ``A fast quantum mechanical algorithm for database search,'' in
  \emph{Proceedings of the twenty-eighth annual ACM symposium on Theory of
  computing}, ser. STOC '96.\hskip 1em plus 0.5em minus 0.4em\relax New York,
  NY, USA: ACM, 1996, pp. 212--219. [Online]. Available:
  \url{http://doi.acm.org/10.1145/237814.237866}
\BIBentrySTDinterwordspacing

\bibitem{Mos09}
M.~Mosca, ``Quantum algorithms,'' in \emph{Encyclopedia of Complexity and
  Systems Science}, 2009, pp. 7088--7118.

\bibitem{CW10}
\BIBentryALTinterwordspacing
A.~M. Childs and W.~van Dam, ``Quantum algorithms for algebraic problems,''
  \emph{Rev. Mod. Phys.}, vol.~82, pp. 1--52, Jan 2010. [Online]. Available:
  \url{http://link.aps.org/doi/10.1103/RevModPhys.82.1}
\BIBentrySTDinterwordspacing

\bibitem{Calderbank1996}
\BIBentryALTinterwordspacing
A.~R. Calderbank and P.~W. Shor, ``Good quantum error-correcting codes exist,''
  \emph{Phys. Rev. A}, vol.~54, pp. 1098--1105, Aug 1996. [Online]. Available:
  \url{http://link.aps.org/doi/10.1103/PhysRevA.54.1098}
\BIBentrySTDinterwordspacing

\bibitem{Steane1996}
\BIBentryALTinterwordspacing
A.~M. Steane, ``Simple quantum error-correcting codes,'' \emph{Phys. Rev. A},
  vol.~54, pp. 4741--4751, Dec 1996. [Online]. Available:
  \url{http://link.aps.org/doi/10.1103/PhysRevA.54.4741}
\BIBentrySTDinterwordspacing

\bibitem{Laflamme1996}
\BIBentryALTinterwordspacing
R.~Laflamme, C.~Miquel, J.~P. Paz, and W.~H. Zurek, ``Perfect quantum error
  correcting code,'' \emph{Phys. Rev. Lett.}, vol.~77, pp. 198--201, Jul 1996.
  [Online]. Available: \url{http://link.aps.org/doi/10.1103/PhysRevLett.77.198}
\BIBentrySTDinterwordspacing

\bibitem{Gottesman1996}
\BIBentryALTinterwordspacing
D.~Gottesman, ``Class of quantum error-correcting codes saturating the quantum
  hamming bound,'' \emph{Phys. Rev. A}, vol.~54, pp. 1862--1868, Sep 1996.
  [Online]. Available: \url{http://link.aps.org/doi/10.1103/PhysRevA.54.1862}
\BIBentrySTDinterwordspacing

\bibitem{Knill1996}
E.~{Knill} and R.~{Laflamme}, ``Concatenated quantum codes,'' \emph{eprint
  arXiv:quant-ph/9608012}, Aug. 1996.

\bibitem{Grassl2009}
M.~Grassl, P.~Shor, and B.~Zeng, ``Generalized concatenation for quantum
  codes,'' in \emph{Information Theory, 2009. ISIT 2009. IEEE International
  Symposium on}, 28 2009-july 3 2009, pp. 953 --957.

\bibitem{Grassl2009a}
\BIBentryALTinterwordspacing
M.~Grassl, P.~Shor, G.~Smith, J.~Smolin, and B.~Zeng, ``Generalized
  concatenated quantum codes,'' \emph{Phys. Rev. A}, vol.~79, p. 050306, May
  2009. [Online]. Available:
  \url{http://link.aps.org/doi/10.1103/PhysRevA.79.050306}
\BIBentrySTDinterwordspacing

\bibitem{Beigi2011}
\BIBentryALTinterwordspacing
S.~Beigi, I.~Chuang, M.~Grassl, P.~Shor, and B.~Zeng, ``Graph concatenation for
  quantum codes,'' \emph{Journal of Mathematical Physics}, vol.~52, no.~2, p.
  022201, 2011. [Online]. Available:
  \url{http://link.aip.org/link/?JMP/52/022201/1}
\BIBentrySTDinterwordspacing

\bibitem{Chuang2009}
\BIBentryALTinterwordspacing
I.~Chuang, A.~Cross, G.~Smith, J.~Smolin, and B.~Zeng, ``Codeword stabilized
  quantum codes: Algorithm and structure,'' \emph{Journal of Mathematical
  Physics}, vol.~50, no.~4, p. 042109, 2009. [Online]. Available:
  \url{http://link.aip.org/link/?JMP/50/042109/1}
\BIBentrySTDinterwordspacing

\bibitem{Chen2008}
\BIBentryALTinterwordspacing
X.~Chen, B.~Zeng, and I.~L. Chuang, ``Nonbinary codeword-stabilized quantum
  codes,'' \emph{Phys. Rev. A}, vol.~78, p. 062315, Dec 2008. [Online].
  Available: \url{http://link.aps.org/doi/10.1103/PhysRevA.78.062315}
\BIBentrySTDinterwordspacing

\bibitem{Cross2008}
A.~Cross, G.~Smith, J.~Smolin, and B.~Zeng, ``Codeword stabilized quantum
  codes,'' in \emph{Information Theory, 2008. ISIT 2008. IEEE International
  Symposium on}, july 2008, pp. 364 --368.

\bibitem{Smolin2007}
\BIBentryALTinterwordspacing
J.~A. Smolin, G.~Smith, and S.~Wehner, ``Simple family of nonadditive quantum
  codes,'' \emph{Phys. Rev. Lett.}, vol.~99, p. 130505, Sep 2007. [Online].
  Available: \url{http://link.aps.org/doi/10.1103/PhysRevLett.99.130505}
\BIBentrySTDinterwordspacing

\bibitem{Yu2008}
\BIBentryALTinterwordspacing
S.~Yu, Q.~Chen, C.~H. Lai, and C.~H. Oh, ``Nonadditive quantum error-correcting
  code,'' \emph{Phys. Rev. Lett.}, vol. 101, p. 090501, Aug 2008. [Online].
  Available: \url{http://link.aps.org/doi/10.1103/PhysRevLett.101.090501}
\BIBentrySTDinterwordspacing

\bibitem{Li2010}
\BIBentryALTinterwordspacing
Y.~Li, I.~Dumer, M.~Grassl, and L.~P. Pryadko, ``Structured error recovery for
  code-word-stabilized quantum codes,'' \emph{Phys. Rev. A}, vol.~81, p.
  052337, May 2010. [Online]. Available:
  \url{http://link.aps.org/doi/10.1103/PhysRevA.81.052337}
\BIBentrySTDinterwordspacing

\bibitem{Ashikhmin2001}
A.~Ashikhmin and E.~Knill, ``Nonbinary quantum stabilizer codes,''
  \emph{Information Theory, IEEE Transactions on}, vol.~47, no.~7, pp. 3065
  --3072, nov 2001.

\bibitem{Ketkar2006}
A.~Ketkar, A.~Klappenecker, S.~Kumar, and P.~Sarvepalli, ``Nonbinary stabilizer
  codes over finite fields,'' \emph{Information Theory, IEEE Transactions on},
  vol.~52, no.~11, pp. 4892 --4914, nov. 2006.

\bibitem{Hu2008}
\BIBentryALTinterwordspacing
D.~Hu, W.~Tang, M.~Zhao, Q.~Chen, S.~Yu, and C.~H. Oh, ``Graphical nonbinary
  quantum error-correcting codes,'' \emph{Phys. Rev. A}, vol.~78, p. 012306,
  Jul 2008. [Online]. Available:
  \url{http://link.aps.org/doi/10.1103/PhysRevA.78.012306}
\BIBentrySTDinterwordspacing

\bibitem{Knill1997}
\BIBentryALTinterwordspacing
E.~Knill and R.~Laflamme, ``Theory of quantum error-correcting codes,''
  \emph{Phys. Rev. A}, vol.~55, pp. 900--911, Feb 1997. [Online]. Available:
  \url{http://link.aps.org/doi/10.1103/PhysRevA.55.900}
\BIBentrySTDinterwordspacing

\bibitem{Grassl2008a}
M.~Grassl and M.~Rotteler, ``Quantum goethals-preparata codes,'' in
  \emph{Information Theory, 2008. ISIT 2008. IEEE International Symposium on},
  july 2008, pp. 300 --304.

\bibitem{Grassl2008}
------, ``Non-additive quantum codes from goethals and preparata codes,'' in
  \emph{Information Theory Workshop, 2008. ITW '08. IEEE}, may 2008, pp. 396
  --400.

\bibitem{Li2010a}
\BIBentryALTinterwordspacing
Y.~Li, I.~Dumer, and L.~P. Pryadko, ``Clustered error correction of
  codeword-stabilized quantum codes,'' \emph{Phys. Rev. Lett.}, vol. 104, p.
  190501, May 2010. [Online]. Available:
  \url{http://link.aps.org/doi/10.1103/PhysRevLett.104.190501}
\BIBentrySTDinterwordspacing

\bibitem{Grassl1997}
M.~Grassl and T.~Beth, ``A note on non-additive quantum codes,'' \emph{eprint
  arXiv:quant-ph/9703016}, Mar. 1997.

\end{thebibliography}

\bibliographystyle{IEEEtran}

\end{document}